\documentclass[cleveref,thm-restate]{lipics-v2021}

\pdfoutput=1
\hideLIPIcs
\nolinenumbers

\graphicspath{{./figures}}

\usepackage{xspace}

\makeatletter
\newcommand{\labelword}[2]{%
  \phantomsection
  #1\def\@currentlabel{\unexpanded{#1}}\label{#2}%
}
\makeatother

\newcommand{\lipicsenumref}[1]{\textcolor{lipicsGray}{\sffamily\bfseries(\labelcref{#1})}}

\newcommand{\etal}{\textit{et~al.}\xspace}

\newcommand{\f}{Fr\'echet\xspace}
\newcommand{\dF}{\ensuremath{\overline{d}_\mathrm{F}}}

\newcommand{\bigO}{\ensuremath{\mathcal{O}}}
\newcommand{\from}{\colon}

\newcommand{\R}{\ensuremath{\mathbb{R}}}
\newcommand{\F}{\ensuremath{\mathcal{F}}}
\newcommand{\calP}{\ensuremath{\mathcal{P}}}
\newcommand{\calR}{\ensuremath{\mathcal{R}}}
\newcommand{\NN}{\ensuremath{\mathit{NN}}}

\newcommand{\Cost}{\ensuremath{\mathrm{Cost}}}

\title{A near-linear time exact algorithm for the $L_1$-geodesic Fr\'echet distance between two curves on the boundary of a simple polygon}
\titlerunning{$L_1$-geodesic Fr\'echet distance between curves on the boundary of a simple polygon}

\author
{Thijs van der Horst}
{Department of Information and Computing Sciences, Utrecht University, the Netherlands
\and
Department of Mathematics and Computer Science, TU Eindhoven, the Netherlands}
{t.w.j.vanderhorst@uu.nl}
{https://orcid.org/0009-0002-6987-4489}
{}

\author
{Marc van Kreveld}
{Department of Information and Computing Sciences, Utrecht University, The Netherlands}
{m.j.vankreveld@uu.nl}
{https://orcid.org/0000-0001-8208-3468}
{}

\author
{Tim Ophelders}
{Department of Information and Computing Sciences, Utrecht University, the Netherlands
\and
Department of Mathematics and Computer Science, TU Eindhoven, the Netherlands}
{t.a.e.ophelders@uu.nl}
{https://orcid.org/0000-0002-9570-024X}
{partially supported by the Dutch Research Council (NWO) under project no.\ VI.Veni.212.260.}

\author
{Bettina Speckmann}
{Department of Mathematics and Computer Science, TU Eindhoven, The Netherlands}
{b.speckmann@tue.nl}
{https://orcid.org/0000-0002-8514-7858}
{}

\authorrunning{T. van der Horst, M. van Kreveld, and T. Ophelders, and B. Speckmann}
\Copyright{Thijs van der Horst, Marc van Kreveld, Tim Ophelders, and Bettina Speckmann}

\ccsdesc[100]{Theory of computation~Computational Geometry}

\keywords{Fr\'echet distance, geodesic, simple polygon}

\begin{document}
\maketitle

\begin{abstract}
Let $P$ be a polygon with $k$ vertices. Let $R$ and $B$ be two simple, interior disjoint curves on the boundary of $P$, with $n$ and $m$ vertices. We show how to compute the \f distance between $R$ and $B$ using the geodesic $L_1$-distance in $P$ in $\bigO(k \log nm + (n+m) (\log^2 nm \log k + \log^4 nm))$ time.
\end{abstract}

\section{Introduction}

In a seminal paper from 1992, Alt and Godau~\cite{alt92measuring_resemblance} introduced the \emph{\f distance} to the algorithms community as a better similarity measure for two curves than, for example, the Hausdorff distance. For two polygonal curves with $n$ and $m$ vertices, they addressed the algorithmic problem of computing the \f distance and presented an $O(nm\log nm)$ time algorithm~\cite{alt95continuous_frechet}, which is near-quadratic in the input size.

Since then, research efforts have developed globally in two directions: examining variations or generalizations that can still be solved in near-quadratic time, and restrictions or approximations to beat the quadratic time upper bound.
In 2014, Bringmann~\cite{bringmann14hardness} showed that the \f distance cannot be computed to within a factor $1.001$ in strongly subquadratic time, unless the strong exponential time hypothesis fails.
This result was improved in 2019 by Buchin, Ophelders and Speckmann~\cite{buchin19seth_says}, who showed that this conditional lower bound applies even for curves in 1-dimensional space and for approximations better than a factor~3. On the other hand, Driemel, Har-Peled and Wenk~\cite{driemel12realistic} showed that a near-linear time $(1+\varepsilon)$-approximation exists for curves in the plane that satisfy a sparseness condition. A full overview of the results is beyond the scope of this paper.

One of the variations that has been examined concerns computing the \f distance between two curves in a polygonal environment. Intuitively, the \f distance is based on matching points on the two curves in a continuous and forward manner, and considering the Euclidean distance between matched points. In a polygonal environment, the Euclidean distance is replaced by the Euclidean \emph{geodesic} distance, the length of a shortest path in the environment.
In the setting of two curves inside a simple polygon, Cook and Wenk~\cite{cook10geodesic_frechet} showed that the \f distance can still be computed in near-quadratic time.
Recently, it was shown that if the two curves lie on the boundary of a simple polygon, then a $(1+\varepsilon)$-approximation exists that uses near-linear time~\cite{vanderhorst25geodesic_frechet}. In this paper we show that in the same setting, but using the $L_1$-distance to measure lengths of geodesics rather than the Euclidean distance, an exact algorithm can be given that requires near-linear time.

\subparagraph*{Preliminaries.}
    Let $P$ be a simple polygon.
    A \emph{curve} on the boundary of $P$ is a piecewise-linear function $F \from [0, 1] \to \partial P$ connecting a sequence of \emph{vertices}.
    Consecutive vertices are connected by a straight-line \emph{edge}.
    The curve is \emph{simple} if the preimage of a point on $\partial P$ is at most an interval.
    We denote by $F[x, x']$ the subcurve of $F$ over the domain $[x, x']$.
    We write $|F|$ to denote the number of vertices of $F$.
    
    A \emph{reparameterization} of $[0, 1]$ is a non-decreasing surjection $f \from [0, 1] \to [0, 1]$.
    Two reparameterizations $f$ and $g$ describe a \emph{matching} $(f, g)$ between two curves $R$ (red) and $B$ (blue), where for every $t \in [0, 1]$ the point $R(f(t))$ is matched to $B(g(t))$.
    The matching $(f, g)$ is said to have \emph{cost}
    \[
        \max_t~\overline{d}( R(f(t)), B(g(t)) ),
    \]
    where in this work, $\overline{d}(p, q)$ measures the length of a shortest path in $P$ from $p$ to $q$ under the $L_1$-norm.
    We call such a path an \emph{$L_1$-geodesic}.
    There are potentially uncountably many $L_1$-geodesics from $p$ to $q$.
    On the other hand, under the $L_2$-norm, $P$ contains a unique shortest path from $p$ to $q$, which we denote by $\pi(p,q)$.
    It turns out that $\pi(p, q)$ is also an $L_1$-geodesic, so we use it as a representative.

    A matching with cost at most $\delta$ is called a \emph{$\delta$-matching}.
    The (continuous) \emph{$L_1$-geodesic \f distance} $\dF(R, B)$ between $R$ and $B$ is the minimum cost over all matchings.
    The corresponding matching is a \emph{\f matching}.

    The \emph{parameter space} of $R$ and $B$ is the axis-aligned rectangle $[0, 1]^2$.
    Any point $(x, y)$ in the parameter space corresponds to the pair of points $R(x)$ and $B(y)$ on the two curves.
    A point $(x, y)$ in the parameter space is \emph{$\delta$-close} for some $\delta \geq 0$ if $\overline{d}(R(x), B(y)) \leq \delta$.
    The \emph{$\delta$-free space} $\F_\delta(R, B)$ of $R$ and $B$ is the subset of $[0, 1]^2$ containing all $\delta$-close points.
    A point $q = (x', y') \in \F_\delta(R, B)$ is \emph{$\delta$-reachable} from a point $p = (x, y)$ if $x \leq x'$ and $y \leq y'$, and there exists a bimonotone path in $\F_\delta(R, B)$ from $p$ to $q$.
    Alt and Godau~\cite{alt95continuous_frechet} observe that there is a one-to-one correspondence between $\delta$-matchings between $R[x, x']$ and $B[y, y']$, and bimonotone paths from $p$ to $q$ through $\F_\delta(R, B)$.
    We abuse terminology slightly and refer to such paths as $\delta$-matchings.

\subparagraph*{Problem statement and results.}
    Our input consists of a simple polygon $P$ with $k$ vertices and curves $R, B \from [0, 1] \to \partial P$ with $n$ and $m$ vertices respectively, on the boundary of $P$.
    We restrict ourselves to inputs of the following type:
    \begin{enumerate}[(i)]
        \item\label{type_1} $R$ and $B$ are oppositely oriented, interior-disjoint, simple curves on the boundary of~$P$.
    \end{enumerate}
    We present an algorithm for computing the \f distance $\dF(R, B)$ between $R$ and $B$, when distances between points are measured as the minimum length of a path in $P$, under the $L_1$-norm.
    For our algorithm, we make two assumptions.
    Firstly, without loss of generality, assume that $R$ is oriented clockwise with respect to $P$, and $B$ counter-clockwise.
    Secondly, we assume general position: no edge of $P$ has slope $\pm 1$, and no two vertices of~$P$ have the same $y$-coordinate.
    We compute the \f distance between $R$ and $B$ under the geodesic $L_1$-distance in $P$ in $\bigO(k \log nm + (n+m) (\log^2 nm \log k + \log^4 nm))$ time.

\subparagraph*{Algorithmic outline.}
    We first outline the \emph{decision algorithm}:
    For an additional parameter $\delta \geq 0$, it reports whether $\dF(R, B) \leq \delta$.
    Recall that a $\delta$-matching between $R$ and $B$ corresponds to a bimonotone path in $\F_\delta(R, B)$ from $(0, 0)$ to $(1, 1)$.
    We use divide-and-conquer to decide whether such a path exists.

    \begin{figure}
        \centering
        \includegraphics{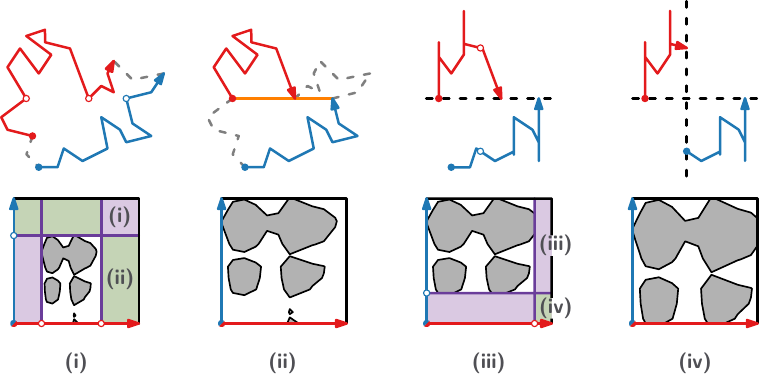}
        \caption{The four types of subproblems that arise in our divide-and-conquer scheme.
        The bottom row shows the parameter spaces corresponding to the top row.
        Between problem types, the correspondence between some subcurves and parts of the $\delta$-free space is illustrated.}
        \label{fig:partition_pipeline}
    \end{figure}
    
    We first explain how to divide the problem of type \lipicsenumref{type_1} into simpler subproblems of four types \lipicsenumref{type_1}--\lipicsenumref{type_4}, illustrated in~\cref{fig:partition_pipeline}.
    In~\cref{sub:partition_general_case} we find a horizontal chord $e^*$ of $P$ that intersects both $R$ and $B$ in at most a constant number of points.
    We split $R$ and $B$ at these points, resulting in subcurves $R_i$ of $R$ and $B_j$ of $B$.
    Each pair $(R_i, B_j)$ corresponds to a simpler subproblem.
    Such a subproblem is either of type \lipicsenumref{type_1} where $R_i$ and $B_j$ have at most a constant fraction of the vertices of $R$ and~$B$, or of a new type with extra structure:
    \begin{enumerate}[(i)]
        \setcounter{enumi}{1}
        \item\label{type_2} $R_i$ and $B_j$ are separated by a horizontal chord of $P$, namely $e^*$.
    \end{enumerate}
    We handle subproblems of type \lipicsenumref{type_1} recursively.

    We handle subproblems of type \lipicsenumref{type_2} in~\cref{sub:separated_polygon}.
    Suppose that $R$ and $B$ are separated by a horizontal chord $e^*$ of $P$, such as the subproblems of type \lipicsenumref{type_2}.
    For a point $p \in P$, denote by $\NN(p)$ the point on $e^*$ closest to $p$ under the $\overline{d}$ metric.
    We leverage the property that for all $r \in R$ and $b \in B$ there exists an $L_1$-geodesic from $r$ to $b$ that goes via $\NN(r)$ and $\NN(b)$.
    Using this property, we transform the problem $(R, B)$ of type \lipicsenumref{type_2} into an equivalent problem of a new type that is no longer constrained by the polygon:
    \begin{enumerate}[(i)]
        \setcounter{enumi}{2}
        \item\label{type_3} a pair of $x$-monotone\footnote{%
            Throughout this work, we consider monotonicity in the weak sense.
            That is, a curve is $x$-monotone if every vertical line intersects the curve in at most one subcurve.
        }
        curves $(\bar{R}, \bar{B})$ that are separated by a horizontal line.
    \end{enumerate}

    We handle subproblems of type \lipicsenumref{type_3} in~\cref{sub:x-monotone}.
    Suppose that $R$ and $B$ are $x$-monotone curves separated by a horizontal line, such as the subproblems of type \lipicsenumref{type_3}.
    We again split the problem into simpler subproblems.
    We find a vertical line $\ell$ that has a constant fraction of the vertices of $R \cup B$ on either side.
    We split $R$ and $B$ at points on $\ell$, resulting in subcurves $R_1$ and $R_2$ of $R$, and $B_1$ and $B_2$ of $B$.
    Each pair $(R_i, B_j)$ corresponds to a simpler subproblem.
    Such a subproblem is either of type \lipicsenumref{type_3} and handled recursively, or of a new type:
    \begin{enumerate}[(i)]
        \setcounter{enumi}{3}
        \item\label{type_4} $R_i$ and $B_j$ are $x$-monotone curves separated by both a horizontal and a vertical line.
    \end{enumerate}

    The above scheme splits our original problem into several subproblems of type \lipicsenumref{type_4}, each of which corresponds to an axis-aligned rectangle in the parameter space of the two input curves.
    Together, these rectangles partition the parameter space into interior-disjoint regions.
    Recall that we want to decide whether there exists a bimonotone path from $(0, 0)$ to $(1, 1)$ in the parameter space that lies completely in the $\delta$-free space.
    To do so, we maintain a set of points in the $\delta$-free space that can be reached by a bimonotone path that starts at $(0, 0)$.
    If for a subproblem of type \lipicsenumref{type_4} we know for its corresponding rectangle a set of reachable ``entry'' points on the bottom and left sides, then we can compute a sufficient set of ``exit'' points on the top and right sides that are reachable via at least one entry point.
    See~\cref{sec:doubly_separated} for details.
    By processing the subproblems of type \lipicsenumref{type_4} in the correct order, we compute whether a bimonotone path from $(0, 0)$ to $(1, 1)$ through the $\delta$-free space exists.

    In~\cref{sub:optimization_algorithm}, we use our decision algorithm to compute $\dF(R, B)$ exactly, by binary searching over a set of $\tilde{\bigO}((n+m)^2)$ candidate values and querying the decision algorithm at every step.
    Computing all candidate values explicitly would take too much time.
    Instead, we represent them implicitly as the Cartesian sum of two sorted sets $X$ and $Y$ of only $\bigO((n+m) \log^2 nm)$ values.
    Using algorithms for linear-time selection in such implicitly represented Cartesian sums~\cite{Frederickson84selection_XY,mirzaian85selection_XY}, the time taken to compute a candidate for $\delta$ is dominated by the time taken to perform the decision algorithm.
    This allows us to compute $\dF(R, B)$ with only logarithmic overhead compared to the decision algorithm.

    We present our decision algorithm bottom-up, starting with the most specific problem \lipicsenumref{type_4} that arises, and ending with the general problem \lipicsenumref{type_1}, in the following sections.

\section{Doubly-separated $x$-monotone curves}
\label{sec:doubly_separated}

    In this section, we consider subproblems of type \lipicsenumref{type_4}, where $R$ and $B$ are two $x$-monotone curves in $\R^2$ separated by both a horizontal line and a vertical line.
    Let $R$ have $n$ vertices and $B$ have $m$ vertices.
    We present an algorithm for \emph{propagating reachability information} through the $\delta$-free space of $R$ and $B$.
    We are given a parameter $\delta \geq 0$, a set of $\bigO(n+m)$ points $S$ on the bottom and left sides of the parameter space, and a set of $\bigO(n+m)$ points $T$ on the top and right sides of the parameter space.
    The goal is to report the points in $T$ that are $\delta$-reachable from at least one point in $S$.
    See~\cref{fig:doubly_separated}.
    In this section, we measure distances under the $L_1$-norm~$\lVert \cdot \rVert_1$, in which shortest paths are not restricted by any polygon.

    \begin{figure}
        \centering
        \includegraphics[page=1]{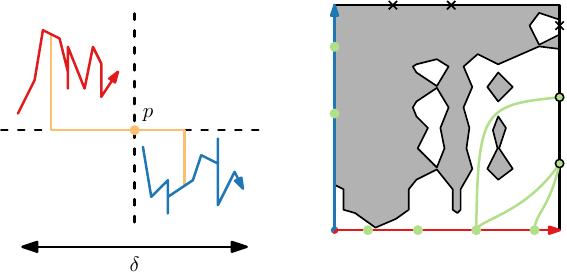}
        \caption{(left) A pair of doubly-separated $x$-monotone curves.
        We may consider the curves as lying on the union of two (degenerate) ``histogram'' polygons, in which case $\overline{d}$ is equal to the $L_1$-norm.
        Between any pair of red and blue points, there exists an $L_1$-geodesic through the point $p$.
        (right) Their $\delta$-free space diagram.
        For propagating reachability, given are a set of points $S$ (disks) on the bottom and left sides of the parameter space, and a set $T$ (circles and crosses) of points on the top and right sides.
        The filled circles are $\delta$-reachable, the crosses are not.}
        \label{fig:doubly_separated}
    \end{figure}

    \Cref{lem:transform_to_separated_1D} shows that the distances between $R$ and $B$ can be captured by two curves in one dimension that are separated by the origin, which can be constructed in linear time.

    \begin{lemma}
    \label{lem:transform_to_separated_1D}
        Let $R$ and $B$ be two doubly-separated curves in $\R^2$ with $n$ and $m$ vertices.
        In $\bigO(n+m)$ time, we can construct two curves $\bar{R}$ and $\bar{B}$ in $\R$ that are separated by the origin, such that $\F_\delta(R, B) = \F_\delta(\bar{R}, \bar{B})$ for all $\delta$.
        The curves $\bar{R}$ and $\bar{B}$ have $n$ and $m$ vertices as well.
    \end{lemma}
    \begin{proof}
        Let $p$ be a point such that the horizontal and vertical lines through $p$ both separate $R$ from $B$.
        Between each pair of points $R(x)$ and $B(y)$, there exists a shortest path (under the $L_1$-norm) through $p$, so $\lVert R(x) - B(y) \rVert_1 = \lVert R(x) - p \rVert_1 + \lVert p - B(y) \rVert_1$.
        Let $\bar{R}, \bar{B} \from [0, 1] \to \R$ where $\bar{R}(x) = -\lVert R(x) - p \rVert_1$ and $\bar{B}(y) = \lVert B(y) - p \rVert_1$ (note the difference in sign).
        This definition lets $p$ correspond to $0$ in one-dimensional space, and places the curve $\bar{R}$ to its left and $\bar{B}$ to its right.
        For any $x,y$, we have $\lVert R(x) - B(y) \rVert_1 = |\bar{B}(y) - \bar{R}(x)|$, so $\F_\delta(R, B) = \F_\delta(\bar{R}, \bar{B})$ for all~$\delta$.
    
        We can compute a horizontal and vertical line separating $R$ and $B$ in $\bigO(n+m)$ time by determining the orthogonal bounding boxes of $R$ and $B$.
        From here we obtain $p$, and constructing $\bar{R}$ and $\bar{B}$ then takes a further $\bigO(n+m)$ time.
    \end{proof}

    Bringmann and K\"unnemann~\cite{bringmann17cpacked} and Van der Horst~\etal~\cite{vanderhorst25geodesic_frechet} both give near-linear time algorithms for propagating reachability information through the free space of two separated curves in one dimension.
    Combining this with~\cref{lem:transform_to_separated_1D}, we obtain:

    \begin{lemma}
    \label{lem:propagating_reachability_doubly_separated}
        Let $R$ and $B$ be two doubly-separated $x$-monotone curves in $\R^2$ with $n$ and $m$ vertices.
        Given a parameter $\delta \geq 0$ and sets $S$ and $T$ of $\bigO(n+m)$ points on the bottom and left, respectively top and right, sides of the parameter space of $R$ and $B$, we can report the points in $T$ that are $\delta$-reachable from at least one point in $S$ in $\bigO((n+m) \log nm)$ time.
    \end{lemma}

\section{Partitioning the parameter space}
\label{sec:partition}

    In this section, we consider the main problem, of type \lipicsenumref{type_1}, where $R$ and $B$ are simple, interior-disjoint curves on the boundary of a simple polygon $P$.
    We divide this problem into subproblems of type \lipicsenumref{type_4} by partitioning the parameter space of $R$ and $B$ into axis-aligned rectangles where either the corresponding subcurves are trivial (i.e., line segments), or form a type \lipicsenumref{type_4} instance, in that the subcurves are equivalent to a pair of doubly-separated $x$-monotone curves.
    
    An axis-aligned rectangle $[x, x'] \times [y, y']$ is \emph{doubly-separated} if, for all $\delta$, the $\delta$-free space $\F_\delta(R[x, x'], B[y, y'])$ inside it is equivalent to the $\delta$-free space of a pair of doubly-separated $x$-monotone curves $\bar{R}$ and $\bar{B}$, i.e., $\F_\delta(R[x, x'], B[y, y']) = \F_\delta(\bar{R}, \bar{B})$.
    A \emph{doubly-separated partition} of the parameter space of $R$ and $B$ is a partition into axis-aligned rectangles $[x, x'] \times [y, y']$, where either $[x, x'] \times [y, y']$ is doubly-separated, or both $R[x, x']$ and $B[y, y']$ are line segments.
    In the latter case, we refer to such a region as \emph{trivial}.
    A doubly-separated partition always exists, since we can take its regions to be only trivial regions, corresponding to the pairs of individual edges of the curves.
    Such a partition has high complexity however, having $nm$ regions.
    Our goal is instead to compute a doubly-separated partition with few trivial regions, and where the doubly-separated regions correspond to subcurves with relatively few vertices in total, so that the total input complexity when applying the result of~\cref{lem:propagating_reachability_doubly_separated} is kept low.

    To measure the ``quality'' of a doubly-separated partition, we define the following \emph{costs}.
    We define the cost $\Cost(\calR)$ of a trivial region to be $\bigO(1)$, and the cost of a doubly-separated region $\calR = [x, x'] \times [y, y']$ as $\big| R[x, x'] \big| + \big| B[y, y'] \big|$, the total number of vertices in the two subcurves.
    We define the cost $\Cost(\calP)$ of a doubly-separated partition $\calP$ as
    $\Cost(\calP) = \sum_{\calR \in \calP} \Cost(\calR)$.
    In this section we construct a doubly-separated partition of only near-linear cost, namely $\bigO((n+m) \log^2 nm)$.

    We construct our partition in three phases.
    In~\cref{sub:x-monotone} we construct a partition for subproblems of type \lipicsenumref{type_3}, where the curves $R$ and $B$ are $x$-monotone and separated by a horizontal line.
    Then, in~\cref{sub:separated_polygon}, we handle subproblems of type \lipicsenumref{type_2}, where $R$ and $B$ are simple curves on the boundary of a simple polygon $P$, separated by a horizontal line segment in $P$.
    Finally, in~\cref{sub:partition_general_case}, we handle subproblems of type \lipicsenumref{type_1}.

\subsection{Vertically-separated $x$-monotone curves}
\label{sub:x-monotone}

    In this section, we consider the setting where $R$ and $B$ are $x$-monotone curves in $\R^2$, separated by a horizontal line.
    We give a simple partition algorithm for the parameter space of $R$ and $B$ that yields a doubly-separated partition of cost $\bigO((n+m) \log nm)$:

    \begin{figure}
        \centering
        \includegraphics[page=2]{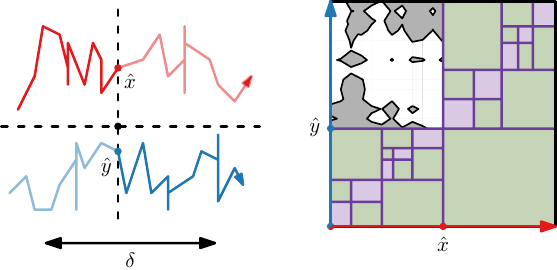}
        \caption{(left) A pair of $x$-monotone curves separated by a horizontal line, together with a vertical line with roughly half of all vertices on either side.
        (right) This line represents a partition of the parameter space into four regions, of which the top-left and bottom-right regions correspond to pairs of doubly-separated subcurves.
        The bottom-left and top-right regions are partitioned further.
        Green regions signify subproblems of type~\lipicsenumref{type_4}.}
        \label{fig:x-monotone}
    \end{figure}

    \begin{lemma}
    \label{lem:partition_x-monotone}
        Let $R$ and $B$ be two vertically-separated $x$-monotone curves in $\R^2$ with $n$ and $m$ vertices.
        We can compute a doubly-separated partition of the parameter space of $R$ and $B$ that has cost $\bigO((n+m) \log nm)$ in $\bigO((n+m) \log nm)$ time.
    \end{lemma}
    \begin{proof}
        We recursively partition the parameter space as follows.
        If $|R| \leq 2$ and $|B| \leq 2$ then the parameter space itself is a doubly-separated partition with only a trivial region and cost $\bigO(1)$, so we stop partitioning further.
        Also, if $R$ and $B$ are already separated by a vertical line, then the parameter space is a doubly-separated region with cost $\bigO(n+m)$, and so we stop partitioning further as well.
        Otherwise, we pick a vertical line $\ell$ and points $R(\hat{x})$ and $B(\hat{y})$ on $\ell$, such that the total number of vertices of $R[0, \hat{x}] \cup B[0, \hat{y}]$ is at most $(n+m) / 2 + 1$, and the total number of vertices of $R[\hat{x}, 1] \cup B[\hat{y}, 1]$ is at most $(n+m) / 2 + 1$ as well.
        The line $\ell$ can be seen as partitioning the parameter space into four quadrants defined by the point $(\hat{x}, \hat{y})$.
        See~\cref{fig:x-monotone}.

        Let $\calP$ be the partition defined by the four quadrants.
        The top-left region in $\calP$ is the parameter space of $R[0, \hat{x}]$ and $B[\hat{y}, 1]$.
        Since these subcurves are separated by $\ell$, as well as by the horizontal line $\ell^*$ separating $R$ and $B$, they are doubly-separated.
        Symmetrically, the subcurves corresponding to the bottom-left region in $\calP$ are doubly-separated as well.
        The bottom-left and top-right quadrants are not doubly-separated; we recursively partition $\calP$ in these regions.

        The cost of $\calP$ is the sum of costs of the partitions in the four quadrants.
        The top-left and bottom-right quadrants are not recursively partitioned and have costs $\big| R[0, \hat{x}] \big| + \big| B[\hat{y}, 1] \big|$ and $\big| R[\hat{x}, 1] \big| + \big| B[0, \hat{y}] \big|$, respectively, totaling at most $n+m+4$.
        The bottom-left and top-right quadrants are partitioned further.
        
        Next we analyze the cost of $\calP$ through a recurrence.
        Let $C(n, m)$ denote the cost of the resulting doubly-separated partition when $R$ and $B$ have $n$ and $m$ vertices, respectively.
        That is, $\Cost(\calP) = C(n, m)$.
        This quantity follows the following recurrence:
        \[
            C(n, m) \leq \begin{cases}
                n + m & \text{if $n \leq 2$ and $m \leq 2$,}\\
                C(n_1, m_1) + C(n_2, m_2) + n+m+4 & \text{otherwise,}
            \end{cases}
        \]
        where $n_1+m_1 \leq (n+m) / 2+2$ and $n_2+m_2 \leq (n+m) / 2+2$.
        This recurrence solves to $\Cost(\calP) = C(n, m) = \bigO((n+m) \log nm)$.

        The line $\ell$ used to split $R$ and $B$ can be found in $\bigO(n+m)$ by walking along the curves, since they are $x$-monotone and thus the vertices are sorted by $x$-coordinate.
        Let $T(n, m)$ denote the time taken to compute the resulting doubly-separated partition when $R$ and $B$ have $n$ and $m$ vertices, respectively.
        This quantity follows the following recurrence:
        \[
            T(n, m) \leq \begin{cases}
                \bigO(1) & \text{if $n \leq 2$ and $m \leq 2$,}\\
                T(n_1, m_1) + T(n_2, m_2) + \bigO(n+m) & \text{otherwise,}
            \end{cases}
        \]
        where again $n_1+m_1 \leq (n+m) / 2+2$ and $n_2+m_2 \leq (n+m) / 2+2$.
        This recurrence solves to $T(n, m) = \bigO((n+m) \log nm)$.
        Thus we compute the partition $\calP$ in $\bigO((n+m) \log nm)$ time.
        Note that our proof easily extends to cases where $R$ and $B$ are not strictly $x$-monotone.
    \end{proof}
    
\subsection{Separated curves on a simple polygon}
\label{sub:separated_polygon}

    Next we reduce the following setting to that of two vertically-separated $x$-monotone curves:
    Let $P$ be a simple polygon, and let $R$ and $B$ be two simple curves on the boundary of $P$, separated by a horizontal line segment in $P$.
    That is, there exists a horizontal line segment $e^*$ in $P$, such that splitting $P$ at $e^*$ creates two subpolygons where $R$ and $B$ lie in different subpolygons.\footnote{%
        We consider $e^*$ to be part of both subpolygons.
    }
    Recall that we assume $R$ to be oriented clockwise with respect to $P$, and $B$ counter-clockwise.
    We construct a pair of vertically-separated $x$-monotone curves whose $\delta$-free space (under the $L_1$-norm) is identical to that of $R$ and $B$, for all $\delta$.

    Every $L_1$-geodesic between points $r \in R$ and $b \in B$ intersects $e^*$ in exactly one line segment.
    Moreover, there exists an $L_1$-geodesic that is composed of the following three parts: an $L_1$-geodesic between $r$ and its closest point $u$ on $e^*$, an $L_1$-geodesic between $b$ and its closest point $v$ on $e^*$, and the line segment $\overline{uv}$.
    For two vertically-separated $x$-monotone curves (as in~\cref{sub:x-monotone}), there exists a similar shortest path between two points, namely one that is composed of two vertical segments and a horizontal segment on the separator.
    We use this connection for our reduction, which we specify next.

    We refer to the ``projection'' of a point $p \in P$ onto $e^*$ (with slope other than $\pm 1$) as the point on $e^*$ closest to $p$, and we denote this projection by $\NN(p)$.
    Above we established that between any two points $r \in R$ and $b \in B$, there exists an $L_1$-geodesic that goes through the projections $\NN(r)$ and $\NN(b)$.
    We introduce two real-valued functions, which we use to define the coordinates of the $x$-monotone curves we reduce the problem to.

    We assume $e^*$ is parameterized over $[1, 2]$.
    The first function, $\varphi \from P \to [1, 2]$, indexes the projections on $e^*$ of the points in $P$.
    That is, $e^*(\varphi(p)) = \NN(p)$.
    The second function, $\psi \from P \to [1, 2]$, represents the distance function from points to their projections, so $\psi(p) = \overline{d}(p, \NN(p))$.
    
    We consider the curves $\bar{R} \from [0, 1] \to \R^2$ and $\bar{B} \from [0, 1] \to \R^2$ given by $\bar{R}(x) = (\varphi(R(x)),\linebreak \psi(R(x)))$ and $\bar{B}(y) = (\varphi(B(y)), -\psi(B(y)))$.
    Note that $\bar{d}(R(x), B(y)) = \lVert \bar{R}(x) - \bar{B}(y) \rVert_1$ for all $x, y \in [0, 1]^2$, and so $\F_\delta(R, B) = \F_\delta(\bar{R}, \bar{B})$ for all $\delta$.
    The curves are separated by the horizontal line $(-\infty, \infty) \times \{0\}$.
    Next we prove that the curves are $x$-monotone, making the algorithm developed in~\cref{sub:x-monotone} applicable to them:

    \begin{lemma}
    \label{lem:monotone_closest_points}
        The curves $\bar{R}$ and $\bar{B}$ are $x$-monotone.
    \end{lemma}
    \begin{proof}
        We prove the statement for $\bar{R}$; it follows by symmetry that $\bar{B}$ is $x$-monotone.
        Let $0 \leq x \leq x' \leq 1$ be two values.
        We show that $\bar{R}(x)$ lies to the left of or on $\bar{R}(x')$.
        
        Let $r = R(x)$ and $r' = R(x')$.
        We have that $\bar{R}(x)$ lies to the left of $\bar{R}(x')$ precisely if $\NN(r)$ comes before $\NN(r')$ along $e^*$.
        Suppose for a contradiction that $\NN(r)$ comes strictly after $\NN(r')$ along $e^*$, for some points $r, r' \in R$ with $r$ before $r'$.
        Due to the orientations of $R$ and $e^*$ relative to each other, every path from $r$ to $\NN(r)$ intersects every path from $r'$ to $\NN(r')$.
        In particular, any pair of rectilinear shortest paths from $r$ and $r'$ to $\NN(r)$ and $\NN(r')$, respectively, intersect in some point $p$.
        Naturally, $\NN(p)$ must be equal to both $\NN(r)$ and $\NN(r')$.
        However, since $\NN(p)$ is unique, we obtain that $\NN(r) = \NN(r')$, giving a contradiction.
    \end{proof}

    We have shown that the algorithm developed in~\cref{sub:x-monotone} can be applied to the curves $\bar{R}$ and $\bar{B}$.
    Next we show that $\bar{R}$ and $\bar{B}$ have low complexity.
    Let $R$ have $n$ vertices and $B$ have $m$ vertices.
    We show that $\bar{R}$ has $\bigO(n)$ vertices and $\bar{B}$ has $\bigO(m)$ vertices.
    While doing so, we develop a data structure that can construct $\bar{R}$ and $\bar{B}$ efficiently.

    We construct pieces of $\bar{R}$ and $\bar{B}$ by computing the functions $\varphi$ and $\psi$ over individual edges of $R$ and $B$.
    When restricted to a single line segment $e$, the function $\psi$ is similar to the ``hourglass function'' $H_{e, e^*}$ introduced by Cook and Wenk~\cite{cook10geodesic_frechet}.
    The hourglass function uses the $L_2$-norm to measure lengths of shortest paths, and can have $\bigO(k)$ complexity~\cite{cook10geodesic_frechet}.
    Perhaps surprisingly however, we show that since we measure lengths with the $L_1$-norm, the function $\psi$ has only constant complexity when applied to a single edge.
    Moreover, the function $\varphi$ has constant complexity as well.

    \begin{figure}
        \centering
        \includegraphics{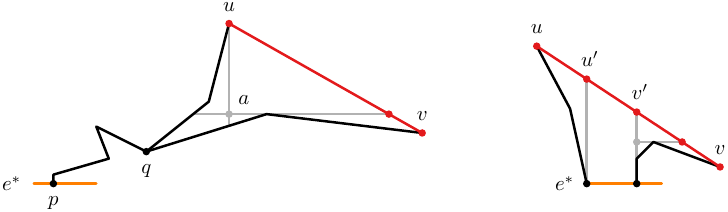}
        \caption{An illustration of the two settings discussed in~\cref{lem:distance_function_edge}.
        On the left is the case where $\NN(u) = \NN(v)$.
        On the right is the case where $\NN(u) \neq \NN(v)$.
        Marked points on $\overline{uv}$ are points where the functions $\varphi$ and $\psi$ change pieces.}
        \label{fig:distance_function_segment}
    \end{figure}

    \begin{lemma}
    \label{lem:distance_function_edge}
        Let $\overline{uv} \subseteq P$ be a line segment interior-disjoint from $e^*$.
        The functions $\varphi$ and $\psi$, restricted to $\overline{uv}$, are piecewise-linear with $\bigO(1)$ pieces.
        After preprocessing $P$ in $\bigO(k)$ time, the functions can be computed over $\overline{uv}$ in $\bigO(\log k)$ time.
    \end{lemma}
    \begin{proof}
        Consider the points $\NN(u)$ and $\NN(v)$.
        We distinguish between the cases where $\NN(u) = \NN(v)$ and $\NN(u) \neq \NN(v)$, proving the complexity bounds and constructing a data structure for each case.
        Computing $\NN(u)$ and $\NN(v)$, and thus determining in which case we are, takes $\bigO(\log k)$ time with the data structure of Cook and Wenk~\cite{cook10geodesic_frechet}, after $\bigO(k)$-time preprocessing.
        For the remainder of this proof, including for the query algorithms of the data structures, we assume these points are known.
        See~\cref{fig:distance_function_segment} for an illustration of the two cases and their proofs.

    \proofsubparagraph{Complexity bound when $\NN(u) = \NN(v)$.}
        We have that $\NN(w) = \NN(u)$ for all $w \in \overline{uv}$.
        Thus $\varphi$ is a constant function when restricted to $\overline{uv}$.
        
        Let $p = \NN(u)$.
        We prove the existence of a point $a \in P$ where for all $w \in \overline{uv}$, $\psi(w) = \lVert w - a \rVert_1 + \overline{d}(a, p)$.
        Since $\lVert w - a \rVert_1$ is a piecewise-linear function with $\bigO(1)$ pieces as we vary $w$ from $u$ to $v$, this point $a$ proves that $\psi$, restricted to $\overline{uv}$, is piecewise-linear with $\bigO(1)$ pieces.
        
        The geodesics $\pi(p, u)$ and $\pi(p, v)$ share a maximal subpath up to some point $q$ (possibly $p$ itself), after which the paths diverge.
        We have $\overline{d}(p, w) = \overline{d}(p, q) + \overline{d}(q, w)$ for all $w \in \overline{uv}$.
        If $q \in \overline{uv}$, then we set $a \gets q$.
        Since the geodesic (under the $L_2$-norm) $\pi(w, a)$ is a line segment for every $w \in \overline{uv}$, we have that $\overline{d}(w, a) = \lVert w - a \rVert_1$, from which it follows that $\psi(w) = \lVert w - a \rVert_1 + \overline{d}(a, p)$.
        
        If $q \notin \overline{uv}$, we consider the simple subpolygon $P'$ bounded by $\pi(q, u)$, $\pi(q, v)$ and $\overline{uv}$.
        The geodesics $\pi(q, u)$ and $\pi(q, v)$ are \emph{inward-convex}~\cite{lee84shortest_paths_rectilinear_barriers}, meaning the interior of every line segment connecting two points on one of the geodesics lies outside of $P'$.
        This implies that the total angle the geodesics turn is strictly smaller than $\pi$.
        Hence there are exactly one maximal horizontal segment $e_H$ and one maximal vertical segment $e_V$ in $P'$ that are incident to $\overline{uv}$ and one of $\pi(q, u)$ and $\pi(q, v)$, while being tangent to the other.
        See~\cref{fig:distance_function_segment}.
        Furthermore, $e_H$ and $e_V$ must intersect in a point, which we set $a$ to.
        Observe that $a$ as chosen meets our requirements.
        Indeed, for every point $w \in \overline{uv}$, there exists an $L_1$-geodesic to $q$ that goes through $a$.
        Additionally, it is bimonotone up to $a$, so its length is $\lVert w - a \rVert_1 + \overline{d}(a, q)$.
        It follows that $\psi(w) = \lVert w - a \rVert_1 + \overline{d}(a, p)$ and has up to three linear pieces for $w \in \overline{uv}$.

    \proofsubparagraph{Data structure when $\NN(u) = \NN(v)$.}
        We preprocess $P$ into two data structures.
        The first is by Bae and Wang~\cite{bae19L_1_shortest_paths}, which, after $\bigO(k)$-time preprocessing, reports the $L_1$-distance between two query points in $P$ in $\bigO(\log k)$ time.
        The second is by Guibas and Hershberger~\cite{guibas89shortest_paths}.
        After $\bigO(k)$-time preprocessing, this data structure allows us to extract $\pi(p, u)$ (or $\pi(p, v)$), represented by a balanced binary search tree $T_u$ (resp.\ $T_v$) storing the edges in order along the path, in $\bigO(\log k)$ time.
        
        By performing two binary searches, one over $T_u$ and one over $T_v$, we can find the horizontal and vertical tangents of $\pi(p, u)$ and $\pi(p, v)$, if they exist, in $\bigO(\log k)$ time.
        From these tangents we determine $e_H$ and $e_V$ in $\bigO(1)$ additional time; if a tangent does not exist, then the respective segment is incident to either $u$ or $v$, and we also determine it in $\bigO(1)$ additional time.

        The intersection between $e_H$ and $e_V$ gives the point $a$.
        We determine $\overline{d}(a, p)$ in $\bigO(\log k)$ time with the data structure of~\cite{bae19L_1_shortest_paths}.
        From here, we determine the restriction of $\psi$ to $\overline{uv}$ in $\bigO(1)$ additional time, using that $\psi(w) = \lVert w - a \rVert_1 + \overline{d}(a, p)$ for all $w \in \overline{uv}$.
        
    \proofsubparagraph{Complexity bound when $\NN(u) \neq \NN(v)$.}
        If $\NN(u) \neq \NN(v)$, then the geodesics $\pi(u, \NN(u))$ and $\pi(v, \NN(v))$ are disjoint.
        Because the geodesics are again inward-convex~\cite{lee84shortest_paths_rectilinear_barriers}, there are points $u', v' \in \overline{uv}$ for which $\overline{u' \NN(u)}$ and $\overline{v' \NN(v)}$ are vertical segments inside $P$.
        Since these segments naturally make a right angle with $e^*$, we have $\NN(u') = \NN(u)$ and $\NN(v') = \NN(v)$.
        The complexity bounds for $\varphi$ and $\psi$ over the segments $\overline{uu'}$ and $\overline{v'v}$ follow from those in the case where $\NN(u) = \NN(v)$.
        For the segment $\overline{u'v'}$, we have that for all $w \in \overline{u'v'}$, the point $\NN(w)$ is the vertical projection of $w$ onto $e^*$.
        Thus over $\overline{u' v'}$, $\varphi$ is a linear function, and $\psi$ measures the length of the vertical segments from points on $\overline{u' v'}$ to $e^*$, making it also a linear function. One of the pieces $\overline{uu'}$ or $\overline{v'v}$ can have one more break, namely the one vertically closer to $e^*$. The break is at the lowest $y$-coordinate of a point on $\overline{uv}$ that can horizontally be connected to $\overline{u'\NN(u)}$ and $\overline{v'\NN(v)}$ while staying in $P$.
        It follows that $\varphi$ and $\psi$ are piecewise-linear over $\overline{uv}$ with $\bigO(1)$ pieces.

    \proofsubparagraph{Data structure when $\NN(u) \neq \NN(v)$.}
        When $\NN(u) \neq \NN(v)$, we can construct $\varphi$ and $\psi$ over $\overline{u'v'}$ in $\bigO(1)$ time, as we may disregard the polygon $P$ and consider only the $L_1$-norm when constructing the functions.
        The functions over the segments $\overline{uu'}$ and $\overline{v'v}$ take an additional $\bigO(\log k)$ time with the data structure for the first case, where $\NN(u) = \NN(v)$.
    \end{proof}

    By querying the data structure of~\cref{lem:distance_function_edge} with every edge of $R$ and $B$ individually and combining the results, we obtain the functions $\varphi$ and $\psi$ over $R$ and $B$ in $\bigO((n+m) \log k)$ time, after $\bigO(k)$-time preprocessing.
    From this, we extract $\bar{R}$ and $\bar{B}$ in $\bigO(n+m)$ additional time.
    The complexities of $\bar{R}$ and $\bar{B}$ are $\bigO(n)$ and $\bigO(m)$, respectively.
    Thus we obtain:

    \begin{lemma}
    \label{lem:transformation_to_x-monotone}
        Let $P$ be a simple polygon with $k$ vertices.
        Let $R$ and $B$ be two simple curves with $n$ and $m$ vertices on the boundary of $P$, separated by a horizontal line segment in $P$.
        After preprocessing $P$ in $\bigO(k)$ time, we can construct a pair of vertically-separated $x$-monotone curves $\bar{R}$ and $\bar{B}$ with $\F_\delta(R, B) = \F_\delta(\bar{R}, \bar{B})$ in $\bigO((n+m) \log k)$ time.
        The curves $\bar{R}$ and $\bar{B}$ have $\bigO(n)$ and $\bigO(m)$ vertices, respectively.
    \end{lemma}

    \begin{corollary}
    \label{cor:partition_separated_polygon}
        Let $P$ be a simple polygon with $k$ vertices.
        Let $R$ and $B$ be two simple curves with $n$ and $m$ vertices on the boundary of $P$, separated by a horizontal line segment in $P$.
        After preprocessing $P$ in $\bigO(k)$ time, we can compute a doubly-separated partition of the parameter space of $R$ and $B$ that has cost $\bigO((n+m) \log nm)$ in $\bigO((n+m) \log nm)$ time.
    \end{corollary}

\subsection{The general case}
\label{sub:partition_general_case}

    In this section, we start out in the general setting, where $R$ and $B$ are simple, interior-disjoint curves on the boundary of a simple polygon $P$, and construct a doubly-separated partition of the parameter space of $R$ and $B$.

    We follow the approach of~\cref{sub:x-monotone}, in that we partition the parameter space recursively, though this time based on horizontal \emph{chords} of $P$.
    A chord is a maximal segment that does not go outside $P$. By our assumptions, a chord can have at most three points in common with the boundary of $P$. If possible, we will use bichromatic chords, which have a point in common with both curves.
    Each such chord $e^*$ splits each curve into at most three subcurves:
    If $e^*$ intersects $R$ only at an endpoint of $R$, the chord ``trivially'' splits $R$ into the curve $R$ itself.
    Otherwise, $R$ is split into the maximal prefix $R[0, x]$ whose intersection with $e^*$ is only the point $R(x)$, the maximal suffix $R[x', 1]$ whose intersection with $e^*$ is only the point $R(x')$, and the maximal subcurve $R[x, x']$ bounded by~$e^*$.
    See~\cref{fig:polygon_chords}.
    The split of $B$ induced by $e^*$ is defined analogously, but if $R$ is split into three subcurves, then $B$ can be split into at most two subcurves (and vice versa).
    A chord corresponds to a partition of the parameter space into the axis-aligned rectangles whose corresponding subcurves are induced by the splits of $R$ and $B$.
    Thus, a chord corresponds to a partition into at most six regions.

    \begin{figure}
        \centering
        \includegraphics{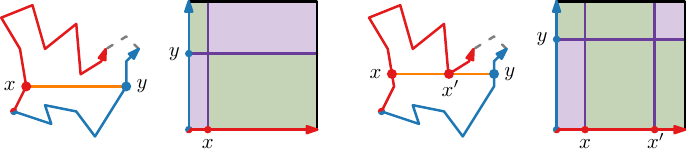}
        \caption{Two types of bichromatic chords and the partitions of the parameter space that they induce.
        All green regions correspond to separated subcurves; the purple regions are split recursively.}
        \label{fig:polygon_chords}
    \end{figure}

    If a bichromatic chord does not exist, then $R$ and $B$ can be separated trivially by a horizontal chord.
    Our partition algorithm works recursively, like the algorithm of~\cref{lem:partition_x-monotone}.
    To bound the recursive depth, we use a chord that splits $R$ and $B$ into at most two and three subcurves each, where the first and last subcurves of each have a total of at most $(n+m)/c$ vertices at either side of the chord, for some constant $c > 1$.
    
    We prove the existence of a horizontal chord with the desired properties, and give a construction algorithm.
    The result follows from a relaxed version of the ``polygon-cutting theorem'' of Chazelle~\cite{chazelle82polygon_cutting}.

    \begin{lemma}
    \label{lem:good_chord}
        If $R$ and $B$ are not already separated by a horizontal chord, there exists a horizontal chord that intersects both $R$ and $B$.
        Furthermore, if $n+m \geq 3$, such a chord exists that splits $R$ and $B$ into at most three subcurves each, where for each pair $(R_i, B_j)$ of subcurves not separated by the chord, $|R_i|+|B_j| \leq 2(n+m)/3$.
        Such a chord can be found in $\bigO(k)$ time.
    \end{lemma}
    \begin{proof}
        If $R$ and $B$ are not already separated by a horizontal chord, then naturally there exists a horizontal chord that intersects both $R$ and $B$.
        For the remainder, assume that $R$ and $B$ are not separated by a horizontal chord.
        
        For the existence of the particular chord we are after, we assume for simplicity that the endpoints of $R$ and $B$ are vertices of $P$, and that these are distinct.
        We make use of the relaxed polygon-cutting theorem of Chazelle~\cite{chazelle82polygon_cutting}.
        To use this theorem, we assign the vertices of $P$ that are also vertices of $R \cup B$ a weight of $1$, and all other vertices a weight of $0$.
        Let $C(P)$ denote the total weight of the vertices of $P$, and note that $C(P) = n+m \geq 3$.
        From Chazelle's relaxed polygon-cutting theorem we obtain the existence of a horizontal chord that splits $P$ into subpolygons $P_1$ and $P_2$ with $C(P_1) \leq C(P_2) \leq 2C(P) / 3$.
        Furthermore, we can compute such a chord $e$ in $\bigO(k)$ time: compute the horizontal decomposition of $P$ with the algorithm of Chazelle~\cite{chazelle91triangulation} and consider the faces in topological order.

        Assuming $P$ is in general position, the chord intersects one curve in at most two points and the other in at most one point.
        Suppose without loss of generality that $e$ intersects $R$ in two points, splitting it into the subcurves $R_1 = R[0, x]$, $R_2 = R[x, x']$ and $R_3 = R[x', 1]$.
        
        If $e$ intersects $B$ in a point $B(y)$, then it splits $B$ into the subcurves $B_1 = B[0, y]$ and $B_2 = B[y, 1]$.
        The prefixes $R_1$ and $B_1$ are not separated by $e$, but they lie in the same subpolygon of $P$.
        Thus $|R_1|+|B_1| \leq \max\{C(P_1), C(P_2)\} \leq 2(n+m)/3$.
        Symmetrically, the suffixes $R_3$ and $B_2$ are not separated by $e$, but satisfy $|R_3|+|B_2| \leq 2(n+m)/3$.
        The other four pairs of subcurves are separated by $e$.

        If $e$ does not intersect $B$, then the endpoints of $e$ are $R(x)$ and $R(x')$.
        We compute a suitable bichromatic chord by sweeping the chord $e$ in the direction away from $R[x, x']$, stretching and shrinking the chord as necessary to remain a chord of $P$.
        We do so until $e$ hits a point on $B$, either with an endpoint of $e$ or internal point.
        Let $e^*$ be this chord.

        The chord $e^*$ splits $R$ into subcurves where the prefix and suffix are strict subcurves of $R_1$ and $R_3$, respectively.
        Because the original chord $e$ does not intersect $B$, we have that $e$ splits $B$ in a trivial way.
        Thus, $B$ lies completely on one side of $e$, and therefore also on one side of $e^*$.
        The number of vertices of prefixes and suffixes induced by the split of $e^*$, on either side of $e^*$, are therefore bounded by those for $e$, and thus by $2(n+m)/3$.

        To compute $e^*$, during the sweep we ensure that the chord always connects a point on $R_1$ and a point on $R_3$.
        This gives a clear choice for how to advance the chord when it reaches a polygon vertex with its interior.
        With the horizontal decomposition of $P$, constructable in $\bigO(k)$ time~\cite{chazelle91triangulation}, this sweepline procedure computes $e^*$ in $\bigO(k)$ time.
    \end{proof}

    With the existence and construction result for good chords at hand, we make an initial partition of the parameter space.
    The partition is not yet a doubly-separated partition, but we refine it afterwards into one that is.

    We recursively partition the parameter space as follows.
    If $|R| \leq 2$ and $|B| \leq 2$ then we stop partitioning further.
    Otherwise, we compute a horizontal chord $e^*$ of $P$ with the algorithm of~\cref{lem:good_chord}, taking $\bigO(k)$ time.

    Assuming $P$ is in general position, the chord intersects one curve in at most two points and the other in at most one point.
    Suppose that $e^*$ intersects $R$ in two points, splitting it into the subcurves $R_1 = R[0, x]$, $R_2 = R[x, x']$ and $R_3 = R[x', 1]$, and that $e^*$ intersects $B$ in one point, splitting it into the subcurves $B_1 = B[0, y]$ and $B_2 = B[y, 1]$.
    These five subcurves together define six axis-parallel rectangles in the parameter space, which is the partition corresponding to $e^*$.

    Let $\calP$ be this partition.
    The subcurve $R_2$ is bounded by $e^*$, and thus is separated from all of $B$ by $e^*$.
    Further, the subcurve $R_1$ is separated from $B_2$, and the subcurve $R_3$ from $B_1$.
    Thus the only regions in the partition whose corresponding subcurves are not separated by a horizontal chord are the bottom-left region $[0, x] \times [0, y]$ and the top-right region $[x', 1] \times [y, 1]$; we recursively partition $\calP$ in these regions.

    We analyze the partition in the same manner as we do for doubly-separated partitions.
    That is, we refer to the cost of an axis-aligned rectangle $\calR = [x, x'] \times [y, y']$ as $\Cost(\calR) = \big| R[x, x'] \big| + \big| B[y, y'] \big|$, the same as for doubly-separated regions, and refer to the cost of the partition $\calP$ as $\Cost(\calP) = \sum_{\calR \in \calP} \Cost(\calR)$.

    \begin{lemma}
    \label{lem:initial_partition}
        The partition $\calP$ has cost $\bigO((n+m) \log nm)$ and can be computed in $\bigO(k \log nm)$ time.
    \end{lemma}
    \begin{proof}
        Let $C(n, m)$ denote the cost of $\calP$ when $R$ and $B$ have $n$ and $m$ vertices, respectively.
        That is, $\Cost(\calP) = C(n, m)$.
        This quantity follows the following recurrence:
        \[
            C(n, m) \leq \begin{cases}
                n + m & \text{if $n \leq 2$ and $m \leq 2$,}\\
                C(n_1, m_1) + C(n_2, m_2) + \bigO(n+m) & \text{otherwise,}
            \end{cases}
        \]
        where $n_1+m_1 \leq 2(n+m)/3+2$ and $n_2+m_2 \leq 2(n+m)/3+2$, as well as $n_1+n_2 \leq n+2$ and $m_1+m_2 \leq m+2$.
        This recurrence solves to $\Cost(\calP) = C(n, m) = \bigO((n+m) \log nm)$.

        The chord $e^*$ is computed in $\bigO(k)$ time (\cref{lem:good_chord}).
        Given $e^*$, the corresponding partition can be computed in $\bigO(n+m)$ time by scanning over the curves.
        This gives a running time of $\bigO(k+n+m)$ per recursive step.
        Naively, since the resulting partition $\calP$ has $\bigO(n+m)$ regions, the total time taken to compute $\calP$ is therefore $\bigO((n+m) \cdot (k+n+m))$.
        Note, however, that when recursively partitioning the region $[0, x] \times [0, y]$, the curves $R_1 = R[0, x]$ and $B_1 = B[0, y]$ lie on the same side of $e^*$.
        That is, $e^*$ splits $P$ into two subpolygons $P_1$ and $P_2$, and $R_1$ and $B_1$ lie in the same subpolygon, say $P_1$.
        Thus, we may restrict ourselves to $P_1$ when computing chords for the further partitioning.
        The same goes for partitioning the region $[x', 1] \times [y', 1]$, though using $P_2$ instead.
        We therefore analyze the time taken to compute $\calP$ through a recurrence.

        Let $T(k, n, m)$ denote the time taken to compute an doubly-separated partition with the above procedure, when $P$ has $k$ vertices and $R$ and $B$ have $n$ and $m$ vertices, respectively.
        That is, the time taken to compute $\calP$ is $T(k, n, m)$.
        This quantity follows the following recurrence:
        \[
            T(k, n, m) \leq \begin{cases}
                \bigO(1) & \text{if $n \leq 2$ and $m \leq 2$,} \\
                T(k_1, n_1, m_1) + T(k_2, n_2, m_2) + \bigO(k+n+m) & \text{otherwise,}
            \end{cases}
        \]
        where $k_1 + k_2 \leq k+4$, as well as $n_1+m_1 \leq 2(n+m)/3+2$ and $n_2+m_2 \leq 2(n+m)/3+2$, and also $n_1+n_2 \leq n+2$ and $m_1+m_2 \leq m+2$.
        This recurrence solves to $T(k, n, m) = \bigO((k+n+m) \log nm) = \bigO(k \log nm)$.
    \end{proof}

    We refine $\calP$ into a doubly-separated partition.
    Some regions in $\calP$ already correspond to line segments on $R$ and $B$, and so are trivial.
    The other regions correspond to subcurves of $R$ and $B$ that are separated by a horizontal chord of $P$.
    We partition these regions further using the result of~\cref{cor:partition_separated_polygon}.

    \begin{lemma}
    \label{lem:partition_polygon_to_x-monotone}
        Let $P$ be a simple polygon with $k$ vertices.
        Let $R$ and $B$ be two simple, interior-disjoint curves on with $n$ and $m$ vertices on the boundary of $P$.
        We can compute a doubly-separated partition of the parameter space of $R$ and $B$ that has cost $\bigO((n+m) \log^2 nm)$ in $\bigO(k \log nm + (n+m) \log^2 nm)$ time.
        Additionally, in the same time bound, we can compute, for each region $[x, x'] \times [y, y']$ in the partition, a pair of doubly-separated $x$-monotone curves $\bar{R}$ and $\bar{B}$ with $\F_\delta(R[x, x'], B[y, y']) = \F_\delta(\bar{R}, \bar{B})$ for all $\delta$.
        The curves $\bar{R}$ and $\bar{B}$ have $\bigO(\big| R[x, x'] \big|)$ and $\bigO(\big| B[y, y'] \big|)$ vertices, respectively.
    \end{lemma}
    \begin{proof}
        Let $\calR = [x, x'] \times [y, y'] \in \calP$ be a region where $R[x, x']$ and $B[y, y']$ are separated by a horizontal chord.
        We partition this region further through the connection of $R[x, x']$ and $B[y, y']$ with vertically-separated $x$-monotone curves.
        For this, we compute a pair of vertically-separated $x$-monotone curves $\bar{R}$ and $\bar{B}$ with $\F_\delta(R[x, x'], B[y, y']) = \F_\delta(\bar{R}, \bar{B})$.
        These curves have $\bigO(\big| R[x, x'] \big|)$ and $\bigO(\big| B[y, y'] \big|)$ vertices, respectively, and can be computed in $\bigO(\Cost(\calR) \log k)$ time, after preprocessing $P$ in $\bigO(k)$ time~(\cref{lem:transformation_to_x-monotone}).

        We compute a doubly-separated partition $\bar{\calP}$ of the parameter space of $\bar{R}$ and $\bar{B}$ that has cost $\bigO(\Cost(\calR) \log \Cost(\calR))$.
        This partition can be computed in $\bigO(\Cost(\calR) \log \Cost(\calR))$ time~(\cref{lem:partition_x-monotone}).
        The partition $\bar{\calP}$ induces a doubly-separated partition of $\calR$ with the same asymptotic cost.

        By partitioning all regions $\calR = [x, x'] \times [y, y'] \in \calP$ where $R[x, x']$ and $B[y, y']$ are separated by a horizontal chord in the above manner, we obtain the desired partition of the parameter space of $R$ and $B$.
        Let $\calP^*$ be the resulting refined partition.
        Its cost is
        \[
            \Cost(\calP^*) = \sum_{\calR \in \calP} \bigO(\Cost(\calR) \log \Cost(\calR)) = \bigO(\Cost(\calP) \log nm) = \bigO((n+m) \log^2 nm).
        \]
        The time taken to compute the initial partition $\calP$ is $\bigO(k \log nm)$~(\cref{lem:initial_partition}).
        The time taken to compute its refinement $\calP^*$, together with the doubly-separated $x$-monotone curves associated with the respective regions, is therefore
        \[
            \bigO(k \log nm) + \sum_{\calR \in \calP} \bigO(\Cost(\calR) \log nm) = \bigO(k \log nm + (n+m) \log^2 nm).\qedhere
        \]
    \end{proof}

\section{Computing the Fr\'echet distance}
\label{sec:frechet_distance}

    In this section, we discuss the main problem setting of this work.
    Let $P$ be a simple polygon and let $R$ and $B$ be two simple, interior-disjoint curves on the boundary of $P$.
    We assume $R$ is oriented clockwise with respect to $P$, and $B$ counter-clockwise.
    We present a near-linear time algorithm for computing $\dF(R, B)$.

    Our algorithm makes use of a \emph{decision} algorithm that, given a parameter $\delta \geq 0$, reports whether $\dF(R, B) \leq \delta$ or $\dF(R, B) > \delta$.
    Recall from the preliminaries that Alt and Godau~\cite{alt95continuous_frechet} showed that $\dF(R, B) \leq \delta$ if and only if there exists a bimonotone path in $\F_\delta(R, B)$ from $(0, 0)$ to $(1, 1)$.
    We develop an algorithm that decides whether such a path exists.
    
    In~\cref{sub:decision_algorithm}, we use the doubly-separated partition constructed in~\cref{sec:partition} to efficiently propagate reachability information through the various regions.
    Recall that there are two types of regions in a doubly-separated partition: trivial regions and doubly-separated regions.
    For trivial regions, we give a simple algorithm for propagating reachability information through them, using that these regions correspond to pairs of segments on $R$ and $B$.
    For the other regions, we use the result of~\cref{sec:doubly_separated} instead, to also propagate reachability information through them.
    This leads to our decision algorithm.
    In~\cref{sub:optimization_algorithm}, we turn our decision algorithm into an optimization algorithm.

\subsection{The decision algorithm}
\label{sub:decision_algorithm}

    We present an algorithm for deciding whether $\dF(R, B) \leq \delta$ for a given parameter $\delta \geq 0$.
    Our algorithm decides whether a bimonotone path in $\F_\delta(R, B)$ exists from $(0, 0)$ to $(1, 1)$.
    Let $\calP^*$ be a doubly-separated partition of the parameter space of $R$ and $B$ computed with the algorithm of~\cref{lem:partition_polygon_to_x-monotone}.
    We propagate reachability information through each region in $\calP^*$ separately.

    Before we do so, we need to address one issue.
    For the regions $\calR \in \calP^*$ where we have a pair of doubly-separated $x$-monotone curves $\bar{R}$ and $\bar{B}$ whose free space is identical to that inside $\calR$, we wish to apply the result of~\cref{lem:propagating_reachability_doubly_separated}.
    However, this algorithm is for propagating reachability from a given discrete set of points to a given discrete set of points.
    Hence we first construct sets $S_\calR$ and $T_\calR$ to which we apply the result of~\cref{lem:propagating_reachability_doubly_separated}.

    We define the sets $S_\calR$ and $T_\calR$.
    For this, we say that a point $B(x)$ is \emph{locally closest} to $R(y)$ if an infinitesimal perturbation of $B(x)$ while staying on $B$ increases its distance to $R(y)$.
    We call points on $R$ locally closest to points on $B$ if a symmetric condition holds.
    The set $T_\calR$ is the set of all points $(x^*, y^*)$ on the top and right sides of $\calR$ for which at least one of $R(x^*)$ and $B(y^*)$ is a vertex of $R$ or $B$, or locally closest to the other point.

    Let $\calR = [x, x'] \times [y, y'] \in \calP^*$.
    The set $T_\calR$ is defined as follows.
    For each edge $e$ of $R$, we add the point $(x^*, y')$ to $T_\calR$ if $x^* \in [x, x']$ and $R(x^*)$ is either a vertex of $R$ or the point on $e$ closest to $B(y')$.
    Symmetrically, for each edge $e$ of $B$, we add the point $(x', y^*)$ to $T_\calR$ if $y^* \in [y, y']$ and $B(y^*)$ is either a vertex of $B$ or the point on $e$ closest to $R(x')$.
    
    The set $S_\calR$ is defined as all points in $S \cap \calR$, as well as all points on the bottom and left sides of $\calR$ that are in a set $T_{\calR'}$ (for an adjacent region $\calR'$) and are $\delta$-reachable from $S_{\calR'}$.
    The set $S_\calR$ is defined as all points on the bottom and left sides of $\calR$ that are in a set $T_{\calR'}$ (for an adjacent region $\calR'$) and are $\delta$-reachable from $S_{\calR'}$.
    If $\calR$ contains the point $(0, 0)$, we set $S_\calR = \{(0, 0)\}$ instead.

    The usefulness of these sets comes from~\cref{lem:locally_closest_matching}, from which it follows that matchings can be assumed to enter and leave a region $\calR$ through points in $S_\calR$ and $T_\calR$, respectively.

    \begin{lemma}
    \label{lem:locally_closest_matching}
        There exists a \f matching between $R$ and $B$ where for every matched pair of points $(r, b)$, at least one of $r$ and $b$ is a vertex or locally closest to the other point.
    \end{lemma}
     \begin{proof}
        Let $(f, g)$ be a \f matching between $R$ and $B$.
        Based on $(f, g)$, we construct a new \f matching $(f', g')$ that satisfies the claim.

        Let $r_1, \dots, r_n$ be the vertices of $R$ and let $b_1, \dots, b_m$ be the vertices of $B$.
        For each $r_i$, if $(f, g)$ matches it to a point interior to an edge $\overline{b_j b_{j+1}}$ of $B$, or to the vertex $b_j$, then we let $(f', g')$ match $r_i$ to the point on $\overline{b_j b_{j+1}}$ closest to it.
        This point is either a vertex or locally closest to $r_i$.
        Symmetrically, for each $b_j$, if $(f, g)$ matches it to a point interior to an edge $\overline{r_i r_{i+1}}$ of $R$, or to the vertex $r_{i+1}$, then we let $(f', g')$ match $b_j$ to the point on $\overline{r_i r_{i+1}}$ closest to it.
        This point is either a vertex or locally closest to $b_j$.
        These matches are monotone, in that if a vertex $r_i$ matches to a point before vertex $b_j$, then $b_j$ matches to a point after $r_i$.

        Consider two maximal subsegments $\overline{r r'}$ and $\overline{b b'}$ of $R$ and $B$ where currently, $r$ is matched to $b$ and $r'$ is matched to $b'$.
        Let $\overline{r r'} \subseteq \overline{r_i r_{i+1}}$ and $\overline{b b'} \subseteq \overline{b_j b_{j+1}}$.
        We have $r = r_i$ or $b = b_j$, as well as $r' = r_{i+1}$ or $b' = b_{j+1}$.

        Let $r^* \in \overline{r r'}$ and $b^* \in \overline{b b'}$ minimize $\overline{d}(r^*, b^*)$.
        It is clear that $r^*$ and $b^*$ are both either vertices or locally closest to the other point.
        We let $(f', g')$ match $r^*$ to $b^*$.
        Since $(f, g)$ originally matched $r^*$ and $b^*$ to points on $\overline{r_i r_{i+1}}$ and $\overline{b_j b_{j+1}}$, respectively, we have $\overline{d}(r^*, b^*) \leq \dF(R, B)$.
        Next we define the part of $(f', g')$ that matches $\overline{r r^*}$ to $\overline{b b^*}$.

        Suppose $r = r_i$ and let $\hat{b} \in \overline{b b^*}$ be the point closest to $r$.
        We let $(f', g')$ match $r$ to $\overline{b \hat{b}}$, and match each point on $\overline{r r^*}$ to its closest point on $\overline{\hat{b} b^*}$.
        This is a proper matching, as the closest point on $\overline{b b'}$ moves continuously along the segment as we move continuously along $\overline{r r'}$.

        The point on $\overline{b \hat{b}}$ furthest from $r$ is $b$, and so the cost of matching $r$ to $\overline{b \hat{b}}$ is $\overline{d}(r, b) \leq \dF(R, B)$.
        For the cost of matching $\overline{r r^*}$ to $\overline{\hat{b} b^*}$, observe that the point on $\overline{\hat{b} b^*}$ closest to a point $\hat{r} \in \overline{r r^*}$ is also the point on $\overline{b_j b_{j+1}}$ closest to it.
        This is due to $\hat{b}$ being closest to $r$ among the points on $\overline{b_j b_{j+1}}$ and $b^*$ being locally closest to $r^*$, which means it is closest to $r^*$ among the points on $\overline{b_j b_{j+1}}$.
        Since $(f, g)$ matches $\overline{r r^*}$ to a subset of $\overline{b_j b_{j+1}}$, it follows that the cost of matching $\overline{r r^*}$ to $\overline{\hat{b} b^*}$ in the way $(f', g')$ does is at most $\dF(R, B)$.

        We define a symmetric matching of cost at most $\dF(R, B)$ between $\overline{r r^*}$ and $\overline{b b^*}$ when $b = b_j$.
        Also, we symmetrically define a matching of cost at most $\dF(R, B)$ between $\overline{r^* r'}$ and $\overline{b^* b'}$.
        The resulting matching $(f', g')$ satisfies the claim.
    \end{proof}

    Naturally, for our decision algorithm, we need to compute these sets efficiently:

    \begin{lemma}
        We can compute the set $T_\calR$ for all $\calR \in \calP^*$ in $\bigO(k + (n+m) \log^2 nm \log k)$ time altogether.
    \end{lemma}
    \begin{proof}
        By~\cref{lem:distance_function_edge}, the function $\varphi$, which measures the distance from points in~$P$ to a given horizontal line segment, is piecewise-linear with $\bigO(1)$ pieces when restricted to a line segment $\overline{uv}$.
        After preprocessing $P$ in $\bigO(k)$ time, we can, given the horizontal line segment and $\overline{uv}$, compute this restriction in $\bigO(\log k)$ time.
        With can use this to compute the point on an edge of $R$ (resp. $B$) that is locally closest to a point $b \in B$ (resp. $r \in R$) in $\bigO(\log k)$ time, by viewing the point as a horizontal line segment.
        Thus, the set $T_\calR$ for a region $\calR = [x, x'] \times [y, y'] \in \calP^*$ can be computed in $\bigO((\big| R[x, x'] \big| + \big| B[y, y'] \big|) \log k) = \bigO(\Cost(\calR) \log k)$ time.
        For all regions in $\calP^*$ together, this takes $\bigO(\Cost(\calP^*) \log k) = \bigO((n+m) \log^2 nm \log k)$ time.
    \end{proof}

    We are now ready to give our decision algorithm, which decides whether a bimonotone path in $\F_\delta(R, B)$ from $(0, 0)$ to $(1, 1)$ exists, for a given $\delta \geq 0$.
    First, we compute the doubly-separated partition $\calP^*$ of~\cref{lem:partition_polygon_to_x-monotone}, taking $\bigO(k \log nm + (n+m) \log^2 nm)$ time.
    The partition has cost $\bigO((n+m) \log^2 nm)$, and for each region $[x, x'] \times [y, y']$ in this partition, either $R[x, x']$ and $B[y, y']$ are line segments, or we have additionally computed a pair of doubly-separated $x$-monotone curves $\bar{R}$ and $\bar{B}$ with $\bigO(\big| R[x, x'] \big|)$ and $\bigO(\big| B[y, y'] \big|)$ vertices, respectively.
    We compute the sets $T_\calR$ for all regions $\calR \in \calP^*$, taking $\bigO(k + (n+m) \log^2 nm \log k)$ time.
    The goal is then to compute, for each set $T_\calR$, the subset of points that are $\delta$-reachable from $(0, 0)$.

    Consider a region $\calR = [x, x'] \times [y, y'] \in \calP^*$ and suppose that for each region $\calR'$ incident to the bottom or left side of $\calR$ we have already computed the subset of $T_{\calR'}$ of points that are $\delta$-reachable from $(0, 0)$.
    The union of these subsets forms the set $S_\calR$ that serves as the set of ``entrances'' for $\calR$; a point in $T_\calR$ is $\delta$-reachable from $(0, 0)$ if and only if it is $\delta$-reachable from a point in $S_\calR$.
    
    If $R[x, x']$ or $B[y, y']$ has more than one edge, we have available a pair of doubly-separated $x$-monotone curves $\bar{R}$ and $\bar{B}$ with $\F_\delta(R[x, x'], B[y, y']) = \F_\delta(\bar{R}, \bar{B})$.
    We apply the algorithm of~\cref{lem:propagating_reachability_doubly_separated} to these curves, using the set $S_\calR$ for the ``entrances'' and the set $T_\calR$ for the ``potential exits.''
    This algorithm reports all points in $T_\calR$ that are $\delta$-reachable from at least one point in $S_\calR$ in $\bigO((|\bar{R}| + |\bar{B}|) \log |\bar{R}| |\bar{B}|) = \bigO(\Cost(\calR) \log \Cost(\calR))$ time.

    If both $R[x, x']$ and $B[y, y']$ are line segments, then we can check in constant time whether a given point in $T_\calR$ is reachable from a given point in $S_\calR$, so we can report the set of points in $T_\calR$ that are $\delta$-reachable from at least one point in $S_\calR$ in $\bigO(|S_\calR| \cdot |T_\calR|) = \bigO(1)$ time.

    \begin{lemma}
    \label{lem:free_space_two_segments}
        If both $R[x, x']$ and $B[y, y']$ are line segments, then a point $t \in T_\calR$ is $\delta$-reachable from a point $s \in S_\calR$ if and only if $t \in \F_\delta(R[x, x']$, $B[y, y'])$ and $t$ lies above and to the right of $s$.
    \end{lemma}
    \begin{proof}
        Let $s = (x_s, y_s) \in S_\calR$ and $t = (x_t, y_t) \in T_\calR$ with $x_t \geq x_s$ and $y_t \geq y_s$.
        We construct a $\delta$-matching between $R[x_s, x_t]$ and $B[y_s, y_t]$.
        To do so, let $r^* \in R[x_s, x_t]$ and $b^* \in B[y_s, y_t]$ minimize $\overline{d}(r^*, b^*)$.
        Naturally, $\overline{d}(r^*, b^*) \leq \overline{d}(R(x_s), B(y_s)) \leq \delta$.
        We construct a $\delta$-matching between $\overline{R(x_s) r^*}$ and $\overline{B(y_s) b^*}$; a $\delta$-matching for the other two subsegments can be constructed symmetrically.

        Let $\hat{b} \in \overline{R(x_s) r^*}$ be the point closest to $R(x_s)$.
        The cost of matching $R(x_s)$ to $\overline{B(y_s) \hat{b}}$ is equal to $\overline{d}(R(x_s), B(y_s)) \leq \delta$.
        From here, we match every point on $\overline{R(x_s) r^*}$ to their closest point on $\overline{B(y_s) b^*}$.
        This is a proper matching, as the closest point varies continuously from $B(y_s)$ to $b^*$ as we vary a point from $R(x_s)$ to $r^*$.
        The cost of matching these pairs of points decreases continuously, thus this part of the matching has cost $\overline{d}(R(x_s) \hat{b}) \leq \delta$.
        The complete matching is therefore a $\delta$-matching.
    \end{proof}

    We now have an algorithm that, given the sets $S_\calR$ and $T_\calR$, reports the points in $T_\calR$ that are $\delta$-reachable in $\bigO(\Cost(\calR) \log \Cost(\calR))$ time.
    Applying this algorithm to all regions in $\calP^*$, in some topological order, we eventually decide whether $(1, 1)$ is $\delta$-reachable from $(0, 0)$, since it is included in a set $T_\calR$.
    This takes
    \[
        \sum_{\calR \in \calP^*} \bigO(\Cost(\calR) \log \Cost(\calR)) = \bigO(\Cost(\calP^*) \log nm) = \bigO((n+m) \log^3 nm)
    \]
    time.
    The computations of $\calP^*$ and the sets $T_\calR$, which took $\bigO(k \log nm + (n+m) \log^2 nm \log k)$ time in total, dominate the running time.
    However, these computations are independent of the decision parameter $\delta$, and so we view these computations as preprocessing.
    This gives the following result:

    \begin{lemma}
    \label{lem:decision_algorithm}
        Let $P$ be a polygon with $k$ vertices.
        Let $R$ and $B$ be two simple, interior disjoint curves on the boundary of $P$, with $n$ and $m$ vertices.
        After preprocessing $P$, $R$ and $B$ in $\bigO(k \log nm + (n+m) \log^2 nm \log k)$ time, we can decide, given $\delta \geq 0$, whether $\dF(R, B) \leq \delta$ in $\bigO((n+m) \log^3 nm)$ time.
    \end{lemma}

\subsection{Obtaining an optimization algorithm}
\label{sub:optimization_algorithm}

    In this section, we turn the decision algorithm of the previous section into an algorithm that computes $\dF(R, B)$, at the cost of a logarithmic factor in the running time when compared to the decision algorithm.
    For this, we derive a polynomial number of candidate distances, one of which is the actual \f distance $\dF(R, B)$.
    Specifically, for two sets of real numbers $X$ and $Y$, their \emph{Cartesian sum}, denoted $X \oplus Y$, is the multiset $X \oplus Y=\{\!\{x+y\mid x\in X, y\in Y\}\!\}$.
    We will show that we can compute two sets $X$ and $Y$ of $\bigO((n+m) \log^2 nm)$ real numbers each, such that $\dF(R,B)\in X \oplus Y$.
    After computing $X$ and $Y$ in near-linear time, we can use existing techniques~\cite{Frederickson84selection_XY,mirzaian85selection_XY} for selection in Cartesian sums to binary search over the Cartesian sum $X \oplus Y$ without explicitly computing its $|X|\cdot|Y|$ many elements.

    We again make use of the doubly-separated partition $\calP^*$ computed with the algorithm of~\cref{lem:partition_polygon_to_x-monotone}.
    We use this partition to identify pairs of separated one-dimensional curves that together describe (almost) all pairwise distances between points on $R$ and points on $B$.
    The sets $X$ and $Y$ will be the set of values of the vertices of these curves.

    First, recall the sets $T_\calR$, for regions $\calR \in \calP^*$, and recall that there exists a \f matching between $R$ and $B$ that corresponds to a bimonotone path in the parameter space that, for each region $\calR \in \calP^*$ that it intersects, goes through a point in $T_\calR$.
    Thus, there exists a pair of points $s \in T_{\calR'}$ and $t \in T_\calR$, for a region $\calR'$ and region $\calR$ incident to the top or right side of $\calR'$, such that the minimum parameter $\delta$ for which there exists a $\delta$-matching from $s$ to $t$ is the \f distance.

    Take regions $\calR'$ and $\calR$, with $\calR$ incident to the top or right side of $\calR'$.
    Let $s \in T_{\calR'}$ and let $t \in T_\calR$ be above and to the right of $s$.
    Let $\delta^*$ be the minimum parameter for which there exists a $\delta^*$-matching from $s$ to $t$.
    We construct sets $X_\calR$ and $Y_\calR$ such that $\delta^*\in X_\calR \oplus Y_\calR$.

    If the subcurves corresponding to $\calR$ both are line segments, then from~\cref{lem:free_space_two_segments} we obtain that $t$ is $\delta$-reachable from $s$, for some $\delta \geq 0$, if and only if both points lie in $\F_\delta(R, B)$.
    We set $X_\calR$ to be the set minimum values $\delta$ for which points $t \in T_\calR$ lie in $\F_\delta(R, B)$.
    The set $Y_\calR$ is simply set to $\{0\}$.
    Taken over all regions in $\calP^*$, the sets have a total cardinality of $\bigO(\Cost(\calP^*)) = \bigO((n+m) \log^2 nm)$.
    
    If one of the subcurves corresponding to $\calR$ has more than one edge, then we have access to a pair of doubly-separated $x$-monotone curves whose $\delta$-free space is identical to the $\delta$-free space of $R$ and $B$ inside $\calR$.
    Moreover, by~\cref{lem:transform_to_separated_1D}, we can transform these curves into a pair of curves in $\R$ that are separated by the point $0$, such that the $\delta$-free space remains the same.
    Let $\bar{R}$ and $\bar{B}$ be these curves.
    Alt and Godau~\cite{alt95continuous_frechet} observe that $\delta^*$ must be one of the following values:
    
    \begin{itemize}
        \item the minimum value for which $s$ and $t$ lie in $\F_{\delta^*}(\bar{R}, \bar{B})$, or

        \item the distance between a vertex of $\bar{R}$ (resp. $\bar{B}$) and an edge of $\bar{B}$ (resp. $\bar{R}$), or

        \item the distance between two vertices of $\bar{R}$ (resp. $\bar{B}$) to an edge of $\bar{B}$ (resp. $\bar{R}$), if the vertices lie at equal distance to the edge.
    \end{itemize}

    Since $\bar{R}$ and $\bar{B}$ are curves in $\R$ that are separated by a point, the candidates for $\delta^*$ of second and third type coincide.
    Moreover, the distance from a vertex $p$ of one curve to an edge $e$ of the other curve is also the distance from $p$ to the endpoint of $e$ closest to $p$, which must be a vertex.
    Thus, the latter two types of candidates for $\delta^*$ can be summarized as all pairwise distances between vertices of $\bar{R}$ and vertices of $\bar{B}$.
    Exploiting the separation of $\bar{R}$ and $\bar{B}$ further, the pairwise distances between vertices of $\bar{R}$ and vertices of $\bar{B}$ can be represented by the Cartesian sum of two sets $\bar{X}$ and $\bar{Y}$, containing the absolute values of the vertex values of $\bar{R}$ and $\bar{B}$, respectively.
    We set $X_\calR \gets \bar{X}$ and $Y_\calR \gets \bar{Y}$.
    Taken over all regions in $\calP^*$, the sets have a total cardinality of $\bigO(\Cost(\calP^*)) = \bigO((n+m) \log^2 nm)$.

    Let $X = \bigcup_{\calR \in \calP^*} X_\calR$ and $Y = \bigcup_{\calR \in \calP^*} Y_\calR$.
    These sets have a total cardinality of $\bigO((n+m) \log^2 nm)$, and contain values $x^* \in X$ and $y^* \in Y$ such that $\dF(R, B) = x^*+y^*$.
    Next we search over $X \oplus Y$ to find the exact value of $\dF(R, B)$.

    After sorting $X$ and $Y$, we can compute the $i^{\mathrm{th}}$ smallest value in $X \oplus Y$, for any given $i$, in $\bigO(|X|+|Y|) = \bigO((n+m) \log^2 nm)$ time~\cite{Frederickson84selection_XY,mirzaian85selection_XY}.
    There are $|X| \cdot |Y|$ values in $X \oplus Y$.
    We binary search over the integers $1, \dots, |X| \cdot |Y|$, and at each considered integer $i$, we compute the $i^{\mathrm{th}}$ smallest value $\delta$ in $X \oplus Y$.
    Then we use our decision algorithm to decide whether $\dF(R, B) \leq \delta$ and to guide the search to the value $\delta^* \in X \oplus Y$ with $\delta^* = \dF(R, B)$.
    
    In the above search procedure, each step takes $\bigO((n+m) \log^3 nm)$ time after preprocessing $P$ and the curves $R$ and $B$ for decision queries, which takes $\bigO(k \log nm + (n+m) \log^2 nm \log k)$ time~(\cref{lem:decision_algorithm}).
    We perform $\bigO(\log (|X| \cdot |Y|)) = \bigO(\log nm)$ steps.
    Thus we obtain our main result:

    \begin{theorem}
        Let $P$ be a polygon with $k$ vertices.
        Let $R$ and $B$ be two simple, interior disjoint curves on the boundary of $P$, with $n$ and $m$ vertices.
        We can compute $\dF(R, B)$ in $\bigO(k \log nm + (n+m) (\log^2 nm \log k + \log^4 nm))$ time.
    \end{theorem}
    
\bibliographystyle{plainurl}
\bibliography{bibliography}

\end{document}